\documentclass[a4paper,USenglish,cleveref,nameinlink,thm-restate]{lipics-v2021}
\usepackage[utf8]{inputenc}

\usepackage{amsmath,amssymb,amsthm}
\usepackage{graphicx}
\usepackage{hyperref}
\usepackage{todonotes}
\usepackage{complexity}
\usepackage{subcaption}
\usepackage{xspace}
\usepackage{enumitem}
\usepackage{thmtools}
\usepackage{apptools}
\usepackage{comment}
\graphicspath{{figures/}}

\definecolor{defblue}{rgb}{0.121,0.47,0.705}
\definecolor{lipicsblue}{rgb}{0.08235294118,0.3098039216,0.537254902}
\let\emph\relax
\DeclareTextFontCommand{\emph}{\color{defblue}\em}
\hypersetup{colorlinks=true,
    linkcolor=lipicsblue,
    anchorcolor=lipicsblue,
    citecolor=lipicsblue,
    filecolor=lipicsblue,
    menucolor=lipicsblue,
    urlcolor=lipicsblue,
    bookmarksopen=true,
    bookmarksopenlevel=2,
    bookmarksnumbered=true,
    plainpages=false,
    }

\newcommand{\restateref}[1]{\IfAppendix{\hyperref[#1]{$\star$}}{\hyperref[#1*]{$\star$}}}

\newcommand{\ver}{arxiv}
\newcommand{\arxapp}[2]{\ifthenelse{\equal{\ver}{arxiv}}{#2}{#1}}

\newtheorem{problem}{Problem}
\newtheorem{property}{Property}

\DeclareMathOperator{\ourfalse}{\texttt{false}}
\DeclareMathOperator{\ourtrue}{\texttt{true}}
\newcommand{\hplus}{\ensuremath{H^+}\xspace}
\newcommand{\hminus}{\ensuremath{H^-}\xspace}

\newcommand{\problemnamek}[1]{$#1$-\textsc{PIP}\xspace}
\newcommand{\problemname}{\problemnamek{1}}

\hideLIPIcs
\nolinenumbers

\renewcommand{\Cref}[1]{\cref{#1}}
\crefname{corollary}{Corollary}{Corollaries}

\title{\boldmath On $k$-Plane Insertion into Plane Drawings}

\titlerunning{\boldmath On $k$-Plane Insertion into Plane Drawings}

\authorrunning{J. Katheder, P. Kindermann, F. Klute, I. Parada, I. Rutter} 
\Copyright{J. Katheder, P. Kindermann, F. Klute, I. Parada, I. Rutter}



\author{Julia Katheder}{Universität Tübingen, Germany}{julia.katheder@uni-tuebingen.de}{https://orcid.org/0000-0002-7545-0730}{Funded by the Deutsche Forschungsgemeinschaft (DFG) -- 364468267.}

\author{Philipp Kindermann}{Universität Trier, Germany}{kindermann@uni-trier.de}{https://orcid.org/0000-0001-5764-7719}{}

\author{Fabian Klute}{Universitat Politècnica de Catalunya, Spain}{fabian.klute@upc.edu}{https://orcid.org/0000-0002-7791-3604}{F. K. is supported by a “María Zambrano grant for attracting international talent” and by grant PID2019-104129GB-I00 funded by MICIU/AEI/10.13039/501100011033.}

\author{Irene Parada}{Department of Mathematics, Universitat Politècnica de Catalunya, Spain}{irene.parada@upc.edu}{https://orcid.org/0000-0003-2401-8670}{I. P. is a Serra H\'unter Fellow. Partially supported by grant 2021UPC-MS-67392 funded by the Spanish Ministry of Universities and the European Union (NextGenerationEU) and by grant PID2019-104129GB-I00 funded by MICIU/AEI/10.13039/501100011033.}

\author{Ignaz Rutter}{University of Passau, Passau, Germany}{rutter@fim.uni-passau.de}{https://orcid.org/0000-0002-3794-4406}{Funded by the Deutsche Forschungsgemeinschaft (DFG) -- 541433306.}

\begin{document}

\maketitle

\begin{abstract}
    We introduce the $k$-Plane Insertion into Plane drawing (\problemnamek{k}) problem: given a plane drawing of a planar graph $G$ and a set $F$ of edges, insert the edges in $F$ into the drawing such that the resulting drawing is $k$-plane. In this paper, we show that the problem
    is \NP-complete for every $k\ge 1$, even when $G$ is biconnected and the set $F$ of edges forms a matching or a path. On the positive side, we
    present a linear-time algorithm for the case that $k=1$ and $G$ is a triangulation.
\end{abstract}

\section{Introduction}

Inserting edges into planar graphs is a long-studied problem in graph drawing.
Most commonly, the goal is to find a way to insert the edges while minimizing the number of
crossings and maintaining the planarity of the prescribed subgraph.
This problem is a core step in the planarization method to find 
graph drawings with few crossings~\cite{DBLP:conf/gd/MutzelZ99}.
Gutwenger et al.~\cite{DBLP:journals/algorithmica/GutwengerMW05} have studied the case of a single edge.
For multiple edges, the picture is more complicated.
In case the edges are all incident to one vertex previously not present in the graph, the problem can be solved in polynomial time~\cite{DBLP:conf/soda/ChimaniGMW09}.
However, the general problem is \NP-hard even when the given drawing is fixed and
the underlying graph is biconnected~\cite{Mutzel1999,Ziegler01}.
Assuming a fixed drawing, Hamm and Hlin\v{e}n\'{y} presented an \FPT-algorithm parameterized by the number of crossings~\cite{hamm_et_al:LIPIcs.SoCG.2022.46}.
Finally, Chimani and Hlin\v{e}n{\'{y}}~\cite{DBLP:journals/jgaa/ChimaniH23} 
gave an \FPT-algorithm for the fixed and variable embedding settings
with the number of inserted edges as a parameter.

In this paper, we take a slightly different viewpoint and
do not restrict the overall number of created crossings,
but instead their structure.
Moreover, we focus on the case when the drawing of the given planar graph $G$ is fixed.
Then our goal is, given a plane drawing $\Gamma$ of $G$
and a set $F$ of edges not present in $G$, 
to find a $k$-plane drawing containing $\Gamma$ as a subdrawing plus the edges of $F$.
Here, a \emph{$k$-plane drawing} of a graph is one in which no edge is crossed
more than $k$ times.
The class of \emph{$k$-planar graphs},
which are those admitting a $k$-plane drawing, is widely studied in graph drawing~\cite{dlm-sgdbp-19,DBLP:books/sp/20/HT2020}.

\begin{figure}[b]
    \centering
    \includegraphics[width=\linewidth]{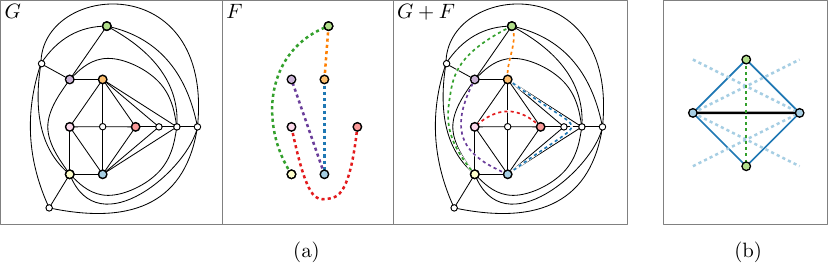}
    \caption{(a) The \problemname problem: a plane graph $G$, a set $F$ of edges, and a 1-plane drawing of $G+F$. (b) In a triangulation, an edge in $G$ (bold) can only be an option for a single edge in $F$ (green) and clashes with at most four other options (blue).}
    \label{fig:teaser}
\end{figure}

\begin{problem}[$k$-Plane Insertion into Plane drawing (\problemnamek{k})]
    Let $\Gamma$ be a plane drawing of a graph $G = (V,E)$ and let $H = (V, E')$ be its complement. Given $G$, $\Gamma$ and
    a set $F \subseteq E'$ of edges,
    find a $k$-plane drawing of the graph $(V, E \cup F)$
    that contains $\Gamma$ as a subdrawing.
\end{problem}

For any fixed $k\in\mathbb N$, 
an \emph{instance} $(G,\Gamma,F)$ of \problemnamek{k} consists of a graph $G$,
a plane drawing $\Gamma$ of $G$, and a set of edges $F$ from the complement of $G$.

\subparagraph*{Our contribution.}
In addition to introducing this problem, we give two results.
In \Cref{sec:algorithm}, we present an $O(|V|)$ algorithm for \problemname
for the case that $G$ is a triangulation.
To accomplish this, we first reduce the number of possible ways one edge can be 
inserted into the given drawing to at most two per edge in $F$ and
then use a $2$-SAT formulation to compute a solution if possible.
In \Cref{sec:reduction}, we show  
that \problemnamek{k} is \NP-complete for every $k\ge 1$ even
if $G$ is biconnected and the edges in $F$ form a path or a matching.

\subparagraph*{Related work.}
\problemnamek{k} is related to the problem of extending a partial drawing of a 
graph to a drawing of the full graph.
Usually, the goal in such problems is to maintain certain properties of the
given drawing.
For example, in works by Angelini et al.~\cite{AngeliniBFJKPR15}, Eiben et al.~\cite{eiben_mfcs_2020,EibenGHKN20}, Ganian et al.~\cite{DBLP:conf/icalp/GanianHKPV21}, or Arroyo et al.~\cite{DBLP:conf/gd/ArroyoDP19,DBLP:journals/dcg/ArroyoKPVSW23} the input is a plane, 1-plane, $k$-plane, or simple drawing, respectively, and 
the desired extension must maintain the property of being plane, 1-plane,  $k$-plane, or simple.
Restrictions of the drawing such as it being straight-line~\cite{ext_straight_06},
level-planar~\cite{PartialConstrainedLevel_Brueckner_2017},
upward~\cite{lozzo_compgeo_2020}, or
orthogonal~\cite{angeliniExtendingPartialOrthogonal2020}
have been explored.
Other results 
consider the number of bends~\cite{cfglms-dpespg-15} or 
assume that the partially drawn subgraph is a cycle~\cite{cegl-dgpwpofpa-12,ExtendingConvexPartial_Mchedlidze_2015}.

\section{\boldmath\problemname: Efficiently inserting edges into a triangulation}
\label{sec:algorithm}

We assume standard notation and concepts from graph theory; compare, e.g.,~\cite{Diestel}.
Given an instance $(G, \Gamma, F)$ of \problemname, often times
an edge $e\in F$ can be inserted into $\Gamma$ in different ways. 
Note that $e$ cannot be inserted without crossings in a triangulation.
An \emph{option} for $e$ is an edge $\gamma$ of $G$ such that $e$ can be inserted into $\Gamma$ crossing only $\gamma$.
Note that in a triangulation, a pair of adjacent faces uniquely 
defines an edge $\gamma$ that must be crossed if $e$ is inserted into said pair of faces. 
Thus, we also use the term option to refer to such a pair of faces. 
An option for $e$ is \emph{safe} if, in case the instance admits a solution, there is a solution in which $e$ is inserted according to this option. 
Two options for two edges $e$ and $e'$ of $F$ \emph{clash} if 
inserting both $e$ and $e'$ according to these options violates 1-planarity.
Examples of safe options are those of edges 
with a single option and an option without clashes. An immediate solution can be found if each edge in $F$ has a non-clashing option. However, it is not sufficient for each edge in $F$ to have a safe option in order to find a solution, e.g., in the case that two single options are clashing.
Observe that in a triangulation, each edge of $\Gamma$ can only be an option for one edge of $F$ and clashes with at most four other options; see~\Cref{fig:teaser}(b).
Further, for a triangulation, we have the following property where a \emph{blocking cycle} in the drawing forces an edge to have only clashing options; see~\Cref{fig:triangulation}(a).
\begin{property} \label{prop:blocking}
    Let $e = (u,v)$ be an edge in $F$ and let $\sigma_i = (x,w)$ be one of its options. For an edge $e'=(w,y) \in F$ having at least one clashing option with $\sigma_i$, there is other non-clashing option with $\sigma_i$, if there is a cycle $C$ in $G$ such that $u,v \in C$ and $x,w,y \notin C$, .
\end{property}
\begin{proof}
    The cycle $C_{\sigma_i} = (u,x,v,w,u)$  and $C$ share only the vertices $u$ and $v$ and since $\Gamma$ is plane, $C$ and $C_{\sigma_i}$ partition $\Gamma$ into four regions, where the edges in $C$ and $C_{\sigma_i}$ constitute the borders of said regions. The edge $e'$ has an option that clashes with $\sigma_i$, i.e., this option is an edge in $C_{\sigma_i}$. Then the endpoint $y$ of $e'$ lies in the region bordered by edges in $C$ and $C_{\sigma_i}$. By $1$-planarity, $e'$ cannot have an option not clashing with $\sigma_i$, as this would require crossing $C$ twice.
\end{proof}

\begin{theorem}
    \problemname can be solved in linear time for instances~$(G,\Gamma,F)$ where~$G$ is a triangulation.    
\end{theorem}

\begin{proof}
    The idea is to preprocess the instance until we are left with a set $F' \subseteq F$ of edges with two options each. 
    The resulting instance can then be solved using a 2SAT formula. 
    We begin by computing all options for every $e\in F$, resulting in $O(|V|)$ options, since each option is an edge in the plane drawing $\Gamma$, crossed by a unique edge in $F$.
    Since $\Gamma$ is plane, we can get the triangles incident to each $v \in V$ in cyclic order and also the options for edges in $F$ incident to $v$. 
    Hence, we get the overall $O(|V|)$ options for edges in $F$ in $O(|V|)$ time.
    For an edge $(u,v) \in F$, $u,v \in V$, with two or more options
    we say that two options are \emph{consecutive} if
    the corresponding faces are consecutive in the cyclic order around $u$ (or $v$); see the options for $(u,v)$ in~\Cref{fig:triangulation}(c) for an illustration.
    We say a set of options is \emph{cyclically consecutive} if 
    the corresponding edges induce a cycle in $G$; see the options for $(u,v)$ in~\Cref{fig:triangulation}(d).
    Whenever an edge $e$ has no option left, we stop and output \texttt{no} and 
    if $e$ has exactly one option left, we insert it into $\Gamma$. 
    Every time we insert an edge, we need to remove at most four options of other edges plus all the options of the just inserted edge. 
    Consider an edge $e = (u,v) \in F$, $u,v \in V$,
    that has three or more options.
    We consider three cases.
    \begin{figure}
    \centering
    \includegraphics[page=1,width=\linewidth]{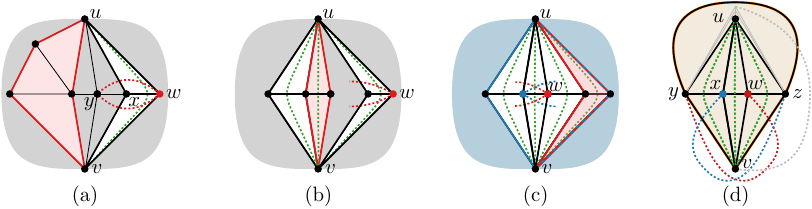}
    \caption{Cases with three or more options in a triangulation.}
    \label{fig:triangulation}
    \end{figure}
    \begin{enumerate}[labelwidth=!, labelindent=0pt, label=(\alph*)]
        \item There are at least three options for $e$, and at least one of them, $\sigma_i$, is not consecutive to any of the other two; see rightmost option in \Cref{fig:triangulation}(b).
        We claim that $\sigma_i$ is either safe or never possible in a solution. 
        If $\sigma_i$ is not clashing with any other option, it is safe and we add it. 
        Otherwise, let $w$ and $x$ be the two endpoints of $\sigma_i$. 
        Option $\sigma_i$ can only be clashing with two options for edges in $F$ incident to $w$ and two options for edges in $F$ incident to $x$. Moreover, any option for those edges clashes with $\sigma_i$. To see this, consider the cycle $C$ formed by $u$, $v$, and the endpoints of another option for $e$ other than $\sigma_i$ (illustrated in red in \Cref{fig:triangulation}(b)). By $\sigma_i$ being is a non-consecutive option, $C$ fulfills~\Cref{prop:blocking} for clashing options of edges in $F$ incident to $x$ and $w$.
        \item There are at least four consecutive non-cyclic options for $e$; see \Cref{fig:triangulation}(c). 
        Let $\sigma_i$ be one of the inner options. Then, similar to the previous case, we can find a blocking cycle as follows. If the option clashes with the rightmost (leftmost) option, we can find a blocking cycle formed by $u$, $v$ and the endpoints of the leftmost (rightmost) option. Otherwise, the cycle formed by $u$, $v$ and the first and last endvertex in the path formed by the consecutive options of $e$ forms a cycle fulfilling~\Cref{prop:blocking} for the endvertices of $\sigma_i$.
        \item There are three consecutive or four cyclically consecutive options for $e$; see \Cref{fig:triangulation}(d). 
        Consider the middle option $\sigma_i$ (or any option if there were four). 
        If it is safe, we just add it. 
        Else, let $w$ and $x$ be the endpoints of $\sigma_i$
        and $y,z$ the other endpoints of options for $e$. 
        Assume, w.l.o.g., that $\sigma_i$ clashes with an option of an edge $e_w$ incident to $w$ and to vertex $y$. 
        For $\sigma_i$ to be a possible option in a solution, $e_w$ must have an option that does not clash with it. 
        There is only one possibility, and it implies that $v,y,z$ or $u,y,z$ form a triangle. 
        Assume, w.l.o.g., the former, so $(y,z)$ is an edge in $\Gamma$.
        Let $V_\diamond$ be the set of vertices $\{u,v,w,x,y,z\}$ and $G_{\diamond}$ the octahedron subgraph of $G$ induced by~$V_\diamond$.

        Edges in $F$ with exactly one endpoint in $V_\diamond\setminus \{u\}$ have at most one option.
        Thus, we can insert them first and see whether we are still in Case~(d). 
        Edges incident to $u$ and to a vertex not in $V_\diamond$ cannot clash with any option of an edge between vertices in $V_\diamond$. 
        Thus, we can solve the constant-size subinstance consisting of inserting such edges into $G_{\diamond}$ independently, taking into account the single-option edges that we might have inserted.
    \end{enumerate}

Once each edge has exactly two options  
we create a 2SAT formula containing
one variable per option and
clauses that ensure exactly one option per edge in $F'$ and exclude clashes. 
This formula has size $O(|V|)$ and
is satisfiable iff the original instance has a solution.
\end{proof}

\section{\texorpdfstring{\boldmath \problemnamek{k_{\ge1}}}{\problemnamek{k}}: Inserting a path or a matching is \texorpdfstring{\NP}{NP}-complete}
\label{sec:reduction}
The membership of \problemnamek{k_{\ge1}} in NP is straightforward; we prove NP-hardness by reduction from \textsc{Planar Monotone 3-SAT}.
Let $\phi$ be a Boolean formula in CNF with variables $V = \{x_1, \dots, x_n\}$ and clauses $C = \{c_1, \dots, c_m\}$. 
Each clause has at most three literals and is either \emph{positive} (all literals are positive) or \emph{negative} (all literals are negative). 
Furthermore, there is a rectilinear representation $\Gamma_\phi$ of the variable-clause incidence graph of $\phi$ such that all variables and clauses are depicted as axis-aligned rectangles or \emph{bars} connected via vertical segments and all variables are positioned on the $x$-axis, all positive clauses lie above, and all negative clauses lie below the $x$-axis; see \Cref{fig:planar-monotone-3-sat} for an example. 
This problem is known to be NP-complete~\cite{DBLP:journals/ijcga/BergK12,DBLP:journals/siamdm/KnuthR92}.
The bars in $\Gamma_\phi$ 
can be layered decreasingly from top to bottom. 
We set the layer of the variables as layer zero, and insert two empty layers, one directly above and one below the variable layer. We denote by $L(c)$ the layer of clause $c$; see~\Cref{fig:planar-monotone-3-sat}.

\begin{figure}[t]
    \centering
    \includegraphics[page=2,scale=0.8]{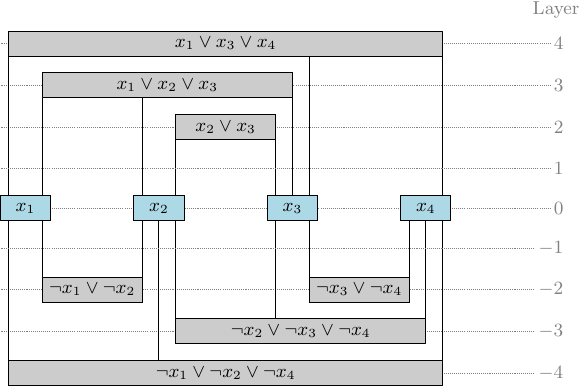}
    \caption{Rectilinear representation of the variable-clause incidence graph of a \textsc{Planar Monotone 3-SAT} instance.}
    \label{fig:planar-monotone-3-sat}
\end{figure}

In the following, starting from $\Gamma_\phi$, we construct a graph $G = (V, E)$, its plane drawing $\Gamma$, and the edge set $F$, which will be inserted into $\Gamma$ in a specific way. We start with the case of $F$ forming a path (see~\Cref{thm:biconnected}) and describe the changes to our construction for $F$ being a matching afterwards (see~\Cref{cor:matching}). 

We denote by $\hplus$ the graph consisting of an axis-aligned $(k+1)\times(k+1)$-vertex grid, where all the $k+1$ vertices on the left side of the grid are connected to a vertex $u$, while all vertices on the right are connected to a vertex $v$. We create chains of copies of $\hplus$, that are connected via the vertices $u,v$. Further, we denote by $\hminus$ the axis-aligned grid graph consisting of $(k-1)\times(k-1)$ vertices. In our construction of $G$, we create grids of $\hminus$ graphs, by connecting two opposing vertical or horizontal sides of their respective vertex grid via $k-1$ non-crossing edges. The grid construction can also be connected to copies of $\hplus$ via $k-1$ non-crossing edges, leaving out the corner vertices of the $\hplus$ vertex grid. If it is necessary to connect a single vertex $v$ to an $\hminus$, we connect $v$ via a fan of $k-1$ edges to one side of the vertex grid. Note that for the case of $k=1$, structures parameterized by $k-1$ such as $\hminus$ are meant to disappear from the construction. \Cref{fig:representations}(a) shows a structure consisting of multiple copies of $\hplus$ and $\hminus$ and their schematic representation used in more complex figures.
We say that an edge $e \in F$ is \emph{$\ell$-spanning} if there are $\ell$ different copies of $\hplus$ in the chain between its endpoints.

\begin{figure}[tbh]
    \centering
    \includegraphics[page=1,width=0.8\linewidth]{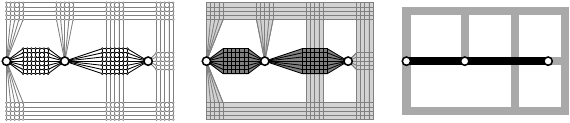}
    \caption{Different representations used in the drawings of our construction. (Left) every vertex and every edge, (middle) a simplification, and (right) a highly abstracted representation.}
    \label{fig:representations}
\end{figure}

\begin{figure}[t]
\centering
\includegraphics[page=3,width=\linewidth]{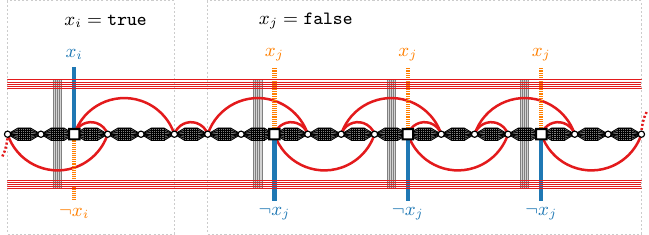}
\caption{Drawing of the variable gadget illustrating~\Cref{lem:variable}.}
\label{fig:variable-gadget}
\end{figure}

\subparagraph{The variable gadget.} 
We replace each bar of a variable $x$ in $\Gamma_\phi$ by a variable 
gadget which consists of an $\hplus$-chain of $4a + 1$ copies,
where $a$ is the maximum over the number of positive and negative occurrences of $x$ in $\phi$.
Let $u_1,\ldots,u_{4a + 2}$ be the vertices that join the copies of $\hplus$ as well as the two unjoined vertices of the first and last $\hplus$ copy in the chain, from left to right.
Moreover, we mark for $i \in \{0,\ldots,a-1\}$ the vertices $u_{4i + 3}$ as
\emph{variable endpoints} (squares in~\Cref{fig:variable-gadget}).
Each such vertex is incident to two \emph{literal edges}, which are connecting the variable gadgets to adjacent layers and
encode the truth value of the respective variables.
We call a literal edge exiting its variable endpoint upwards (downwards)  
\emph{positive} (\emph{negative}). 

For every copy of $\hplus$ with position $4i + 2$, $0 \leq i < a$, we connect the top and the bottom side of its vertex grid to one copy of $\hminus$ each. For each side---top and bottom---of the $\hplus$-chain, the copies of $\hminus$ will be connected by $k-1$ non-crossing edges in $F$ as schematically shown in~\Cref{fig:variable-gadget}. For an illustration showing all vertices and edges, see~\Cref{fig:representations}(a).

For each variable gadget, its edges in $F$ (bold red in~\Cref{fig:variable-gadget}) then consist of alternating 
$3$- and $1$-spanning edges.
Formally, for each $i \in \{0, \dots a-1\}$, the path $F$ passes through the vertices
$u_{4i + 1}$, $u_{4i + 4}$, $u_{4i + 3}$, $u_{4i + 6}$, $u_{4i + 5}$,
except for $a-1$, where we omit the last vertex. 

For the remainder, we depict literal edges representing the value $\ourtrue$ in blue and the ones representing $\ourfalse$ in orange, while the edges in $F$ are colored in red; c.f.~\Cref{fig:variable-gadget,fig:clause-gadget,fig:plane-path,fig:plane-matching}.

\begin{lemma}
    \label{lem:variable}
    Let $v$ be a variable gadget described as above.
    Then, in any $k$-planar drawing containing $v$,
    its literal edges, and the edges $F_v \subseteq F$ incident to vertices in $v$,
    either all negative or all positive literal edges are crossed.
\end{lemma}
\begin{proof}
By $k$-planarity, $F$ cannot cross the $\hplus$-chain. We refer to the $3$-spanning edges that cross literal edges as \emph{blocking $3$-spanning edges} and the others as \emph{forcing $3$-spanning edges}. 
For every blocking $3$-spanning edge $e$, there is a copy of $\hplus$ between its endpoints, which connects to a two $\hminus$ copies, one on the top and one on the bottom side. On their respective side, these $\hminus$ copies are in turn horizontally connected to neighboring $\hminus$ copies by $k-1$ edges in $F$. By $k$-planarity, $e$ cannot cross out of surrounding structure. Hence, $e$ crosses the $k-1$ edges connecting to the $\hplus$ between its endpoints either on the top or bottom side of the gadget. 
Since the blocking $3$-spanning edges are additionally crossed by a literal edge, they cannot cross the forcing $3$-spanning edges (by $k$-planarity) and thus all have to lie on the same side. Subsequently, the forcing $3$-spanning edges all lie on the opposite side.
Hence, in the variable gadget, $F$ crosses either all negated literal 
edges of a variable gadget which corresponds to $x_i$ being $\ourtrue$, or vice versa, corresponding to $x_i$ being $\ourfalse$. 
\end{proof}

We think of the variable corresponding to the gadget 
as set to $\ourtrue$ if the negative literal edges are crossed, and to $\ourfalse$ otherwise.
We connect the variable gadgets by adding one copy of $\hplus$ with a $1$-spanning edge added to $F$ in between them; see \Cref{fig:variable-gadget}.

\begin{figure}[t]
\centering
\includegraphics[page=4,width=\linewidth]{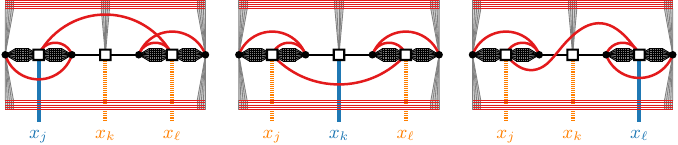}
\caption{Drawing of the clause gadget illustrating~\Cref{lem:clause}.}
\label{fig:clause-gadget}
\end{figure}

\subparagraph{The clause gadget.}  We describe the construction only for the positive clauses; it works symmetrically for the negative ones.
The clause gadget is depicted in~\Cref{fig:clause-gadget}. 
It consists of a chain of two copies of $\hplus$,
followed by two edges,
followed by two more copies of $\hplus$.
We mark the middle vertices of each of the two copies and the two edges as \emph{variable endpoints} and 
add one additional edge to them, their \emph{literal edge}.
Assume that all literal edges are drawn on the same side as shown in~\Cref{fig:clause-gadget} and 
add edges to $F$ as shown in red.
Further, we add three copies of $\hminus$ on the top and two on the bottom side of the gadget, and connect the left- and rightmost vertex in the gadget as well as the middle variable endpoint to the corresponding ones. Similar to the variable gadget, these $\hminus$ copies will be connected via edges in $F$ as shown in~\Cref{fig:clause-gadget}.

\begin{lemma}
    \label{lem:clause}
    Let $c$ be a clause gadget drawn as described above. Then, 
    in any $k$-planar drawing containing $c$, 
    its literal edges, and 
    the edges $F_c \subseteq F$ incident to vertices of $c$,
    at least one literal edge has to be crossed by an edge in $F_c$.
\end{lemma}
\begin{proof}
    We will proceed by showing that in a $k$-planar drawing, at least one literal edge connected to $c$ has to be crossed by an edge in $F_c$.
    Let $v_1, v_2, v_3$ be the variable endpoints of the clause gadget $c$ on layer $i$, 
    $u_1, u_2, u_3$ be the corresponding variable endpoints on layer $i-1$, and 
    $l_j = (u_j, v_j)$ $j \in {1,2,3}$ be the three literal edges connecting them. Observe that $l_1, l_2, l_3$ each gain $k-1$ crossings by crossing the $k-1$ edges in $F$ which are connected to two $\hminus$ copies on the bottom side of $c$. We denote by $e_1, \dots, e_5$ the edges of $F$ following the path in the clause gadget starting from the leftmost vertex.
    Assume for a contradiction that $l_1, l_2, l_3$ represent the value $\ourfalse$.
    Then $l_1$ is crossed by $F$ on layer $i-1$. In order to maintain $k$-planarity, it follows that $e_1$ has to be drawn at the top side of the gadget.
    The next edge $e_2$ spans one $\hplus$ to the left, so $e_3$ has two options. Either it stays the top side of the gadget and crosses $e_1$ or it passes through the bottom side, crossing the left single edge incident to $v_2$ to the top side afterwards, since $l_2$ cannot be crossed. In either case, $e_3$ gains one crossing before having to cross the $k-1$-fan incident to $v_2$. As $e_3$ has $k$ crossings afterwards, it has to connect to $v_3$ on the top side of $c$. However, as $l_3$ is already crossed in layer $i-1$, the edge $e_5$ has to cross to the top side in order to avoid the $k+1$-th crossing at $l_3$. This induces $k+1$ crossings at $e_3$, a contradiction to the $k$-planarity of the drawing. It follows that at most two literal edges can remain uncrossed by edges of $F$ passing through the clause gadget, resulting in three valid states of the gadget, each forcing at least one literal to be $\ourtrue$; see~\Cref{fig:clause-gadget}. 
\end{proof}

\subparagraph{Propagating the variable state.} Again, we only describe the construction for layers $> 0$ as the other side is symmetric. 
We insert $\hplus$-chains with $1$-spanning edges added to $F$ on every layer $> 0$ of $\Gamma_\phi$ and insert the clause gadgets into the respective layers as shown in~\Cref{fig:plane-path}. 
Further, we create variable endpoints on all layers $> 0$ in order to propagate the state of the variable gadgets to clauses in higher layers. 
Layer $1$ thereby ensures that each variable endpoint can be connected to another endpoint, even if the respective literal edge is not used in a clause, as this is crucial to ensure the alternating pattern in the variable gadgets; see, e.g., $x_1$ in~\Cref{fig:plane-path}. 
For each pair of corresponding variable endpoints of a variable gadget and clause gadget, we create a variable endpoint at a merged vertex in the $\hplus$-chain in each layer $i$ with $0 < i < L(c)$ and insert \emph{propagating edges}, by prescribing $F$ to span the two neighboring copies of $\hplus$. 
Further, we connect every two consecutive variable endpoints on layer $j$ and $j+1$ with $1 \leq j < L(c)$ via a literal edge, as illustrated in~\Cref{fig:plane-path}. 

Both the variable and the clause gadget require each literal edge to have either $k-1$ or $k$ crossings. Since the \problemnamek{k_{\ge1}} problem requires $\Gamma$ to be plane, it is not possible in our construction to create these crossings with edges in $G$, hence the path formed by edges $F$ has to cross each literal edge $k-1$ times, in addition to one potential crossing by the propagation edges. To this end, we create a subpath $P_{i}$ comprised of edges in $F$ between each layer $i$ and layer $i-1$ (if present), which is passing through copies of $\hminus$ and crosses the literal edges $k-1$ times; see~\Cref{fig:rep-red-path} for different levels of abstraction used in our illustrations. The subpath of $F$ on every layer $i$ is joined to $P_i$ and $P_{i+1}$ (if present) by an $\hplus$-chain and $1$-spanning edges. For $k=1$ we simply connect the subpaths of $F$ on each layer to the next by a$\hplus$-chain and $1$-spanning edges. Note that if $k$ is even, the sides where the $\hplus$-chain is located alternate, otherwise they connect on the same side of the drawing. Note that in our illustrations showing final the constructions for the graph given in~\Cref{fig:planar-monotone-3-sat}, we assume an even $k$. To ensure that the edges in each $P_i$ do not exceed $k$ crossings, we subdivide between each literal edge by inserting vertically connected copies of $\hminus$ as depicted in~\Cref{fig:plane-path}.

\begin{figure}[tbh]
    \centering
    \includegraphics[page=10,width=0.8\linewidth]{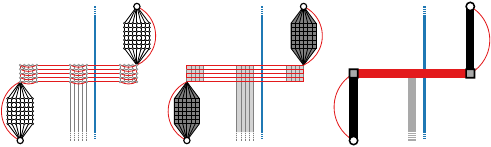}
    \caption{Different representations of alternating path $P_i$ between layer $i$ and layer $i-1$.}
    \label{fig:rep-red-path}
\end{figure}

\begin{lemma}
    \label{lem:propagation}
    Let $P = e_1, \dots, e_{L(c)}$ be a path of literal edges such that 
    $e_1$ is incident to a variable endpoint of variable gadget $v$ and 
    $e_{L(c)}$ to one clause gadget $c$,
    if $e_1$ is crossed in $v$, then
    $e_{L(c)}$ is crossed by an edge of $F$ incident to vertices on layer $L(c)-1$.
\end{lemma}
\begin{proof}
    Let $F_j \subset F$, be the edges in the $k-1$-alternating path between layer $j$ and $j+1$. We first show that the edges in $F_j$ induce $k-1$ crossings on the literal edge $e_{j+1}$. Let $e \in F_j$ be an edge incident to a vertex in the $\hminus$ copy to the left side and to another vertex in the $\hminus$ copy to the right side of $e_{j+1}$. Assume for a contradiction, that $e$ does not cross $e_{j+1}$. It follows that $e$ is drawn as an arc around one of the endpoints of $e_{j+1}$, say $v$, of the literal edge. However, as $v$ lies on a chain of $\hplus$ copies, $e$ necessarily crosses this chain; a contradiction to $k$-planarity. It follows that each literal edge $e_{j+1}$ has exactly $k-1$ crossings with $F_j$. 
    The path $F$ passing through layer $j$ either crosses the literal edge $e_j$ or the literal edge $e_{j+1}$. However, if $e_1$ is crossed by $F$ in the variable gadget, it has $k$ crossings in total, so $F$ is forced to cross $e_2$ passing through layer $1$, propagating the state of the variable gadget to the next layer, such on layer $j$, $F$ has to cross $e_{j+1}$. It follows that every internal vertex of $P$ acts as a variable endpoint for the next layer. Hence, if $e_1$ is crossed by $F$ passing through the variable layer, the edge $e_{L(c)}$ at the clause gadget $c$ is crossed by $F$ in layer $L(c) - 1$.
\end{proof}

Note that the first edge of $P$ being uncrossed in the variable gadget 
does not necessarily lead to its last edge being uncrossed in layer $L(c)-1$.
In fact, this is possible when multiple literals evaluate to $\ourtrue$ for a clause gadget; e.g., the top-left orange edge in \Cref{fig:plane-path}.

\begin{figure}[tbh]
    \centering
    \includegraphics[page=5,width=0.9\linewidth]{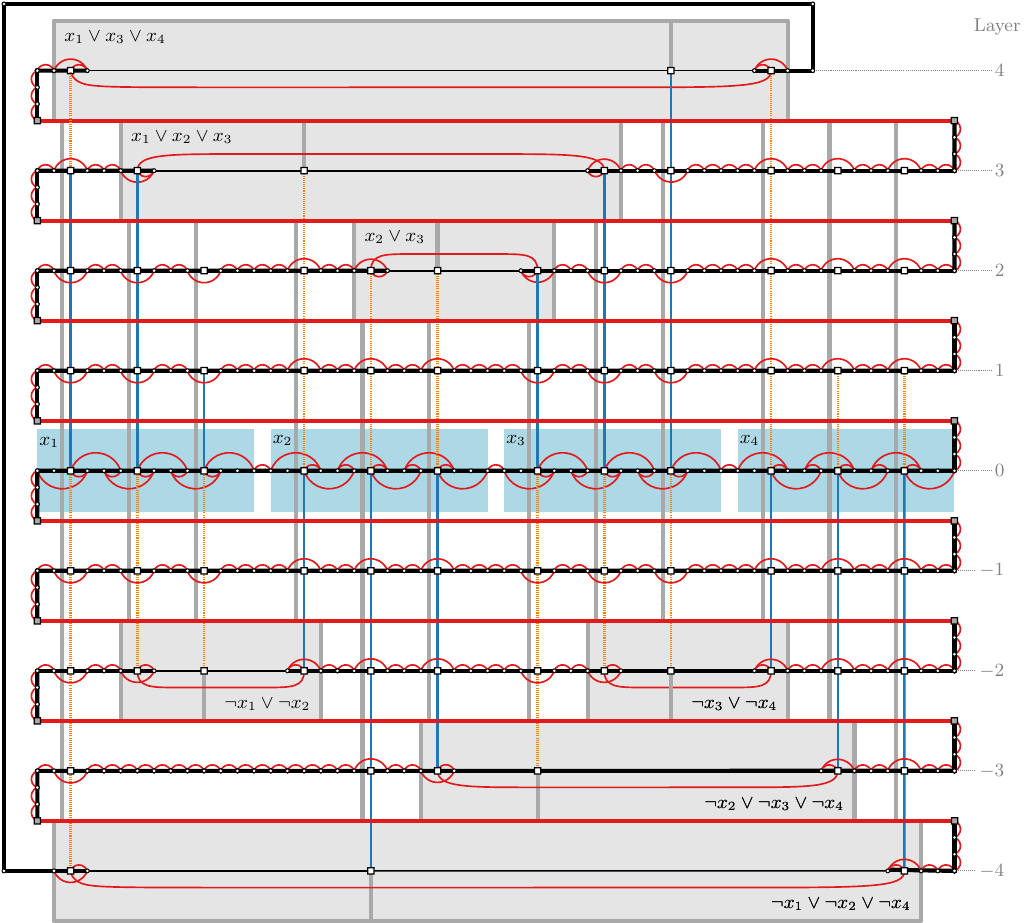}
    \caption{Solution (in red) of the \problemnamek{k} instance coming from the graph given in~\Cref{fig:planar-monotone-3-sat}.}
    \label{fig:plane-path}
\end{figure}

\begin{theorem}\label{thm:biconnected}
    \problemnamek{k} is \NP-complete for every $k\ge 1$, even if $G$ is biconnected and
    $F$ forms a path.
\end{theorem}
\begin{proof}
Assume that $\phi$ is a yes-instance, i.e., there exists an assignment for the variables such that at least one literal in every clause is satisfied. We obtain a $k$-planar drawing of our construction as follows. For every variable gadget, we draw the blocking $3$-spanning edges to be above the gadget if the variable is $\ourfalse$, and below otherwise. On the intermediate layers, we propagate the state of the variable gadget as described.
Since each clause is satisfied, for each clause, one of the three valid states of the clause gadget can be chosen to obtain a $k$-planar drawing. 

Now assume that our construction is a yes-instance, i.e., there exists a $k$-planar drawing of $(V,E\cup F)$  that contains $\Gamma$ as a subdrawing. 
As argued above, at each variable gadget, all blocking $3$-spanning edges must be drawn on the same side. 
We assign each variables a $\ourfalse$ value if the blocking $3$-spanning edges on the corresponding variable gadget are drawn above, and $\ourtrue$ otherwise. 
For the clause gadgets, as argued above, one of the literal edges has to be crossed by the part of $F$ that passes through the gadget. Because of the variable propagation gadget, this is only possible if the literal edge at the corresponding variable gadget is uncrossed by $F$ in the variable gadget; hence, the literal evaluates to $\ourtrue$. Thus, this variable assignment is a feasible solution for $\phi$.
\end{proof}

Note that the construction for \problemname naturally becomes much simpler, as the (gray) axis-aligned grid graph $H^-$ consists of 0 vertices; see \Cref{fig:plane-1-path,fig:plane-1-matching}.

\begin{figure}[tbh]
    \centering
    \includegraphics[page=11,width=0.8\linewidth]{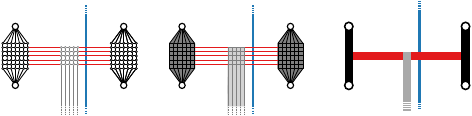}
    \caption{Different representations of the $k-1$ matching edges which substitute the alternating path $P_i$ in the case that $F$ is a matching.}
    \label{fig:rep-red-matching}
\end{figure}
\begin{figure}[tbh]
    \centering
    \includegraphics[page=6,width=0.8\linewidth]{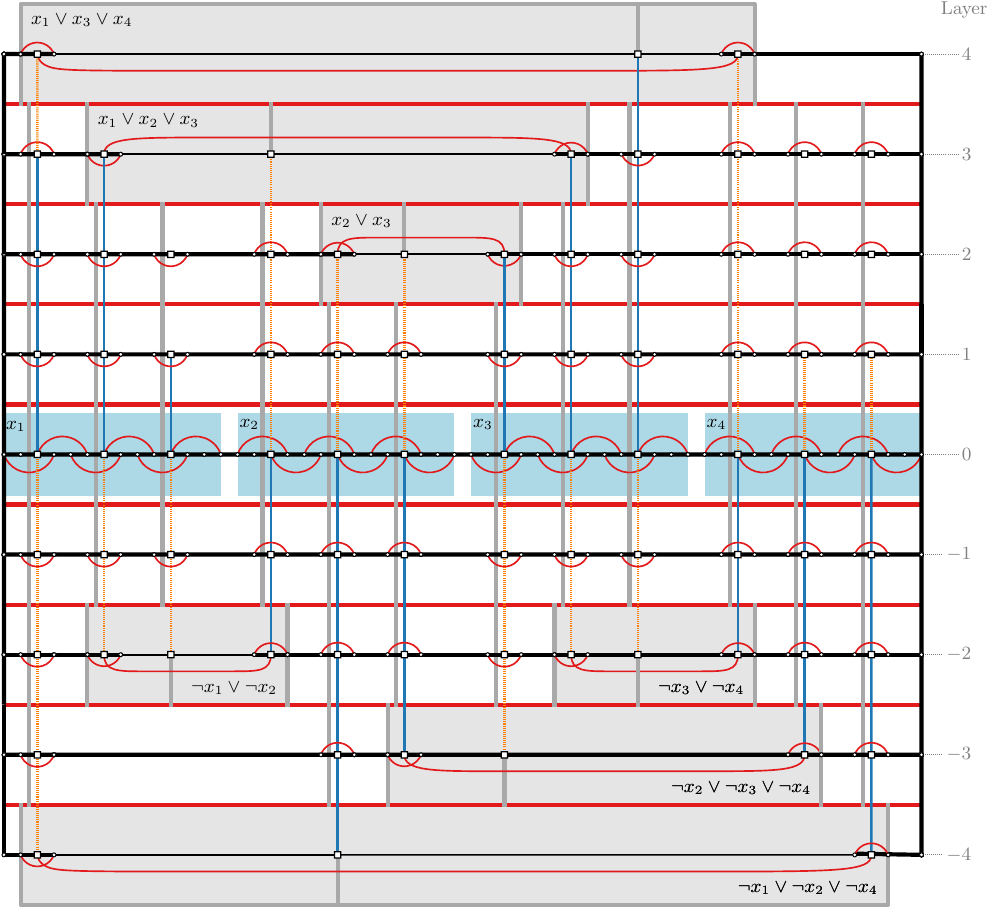}
    \caption{\problemnamek{k} instance where $F$ is a matching reduced from 
 the graph given in~\Cref{fig:planar-monotone-3-sat}}
    \label{fig:plane-matching}
\end{figure}

We can use essentially the same construction, but replace the alternating connections
between the layers by single edges to prove NP-hardness also for the case that $F$ is a matching; 
see \Cref{fig:plane-matching,fig:rep-red-matching}.

\begin{corollary}\label{cor:matching}
    \problemnamek{k} is \NP-complete for every $k\ge 1$, even if $G$ is biconnected and
    $F$ forms a matching.
\end{corollary}

\begin{figure}[t]
    \centering
    \includegraphics[page=8,width=0.8\linewidth]{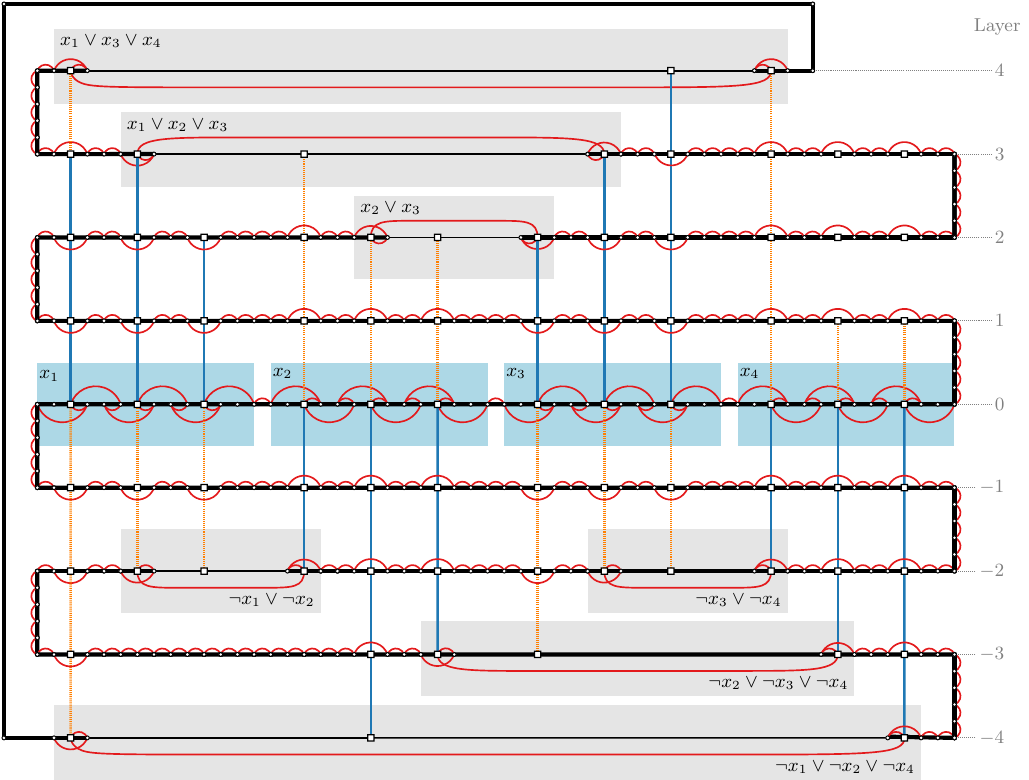}
    \caption{Solution (in red) of the \problemname instance coming from the graph given in~\Cref{fig:planar-monotone-3-sat} for $F$ being a path.}
    \label{fig:plane-1-path}
\end{figure}

\begin{figure}[b]
    \centering
    \includegraphics[page=9,width=0.8\linewidth]{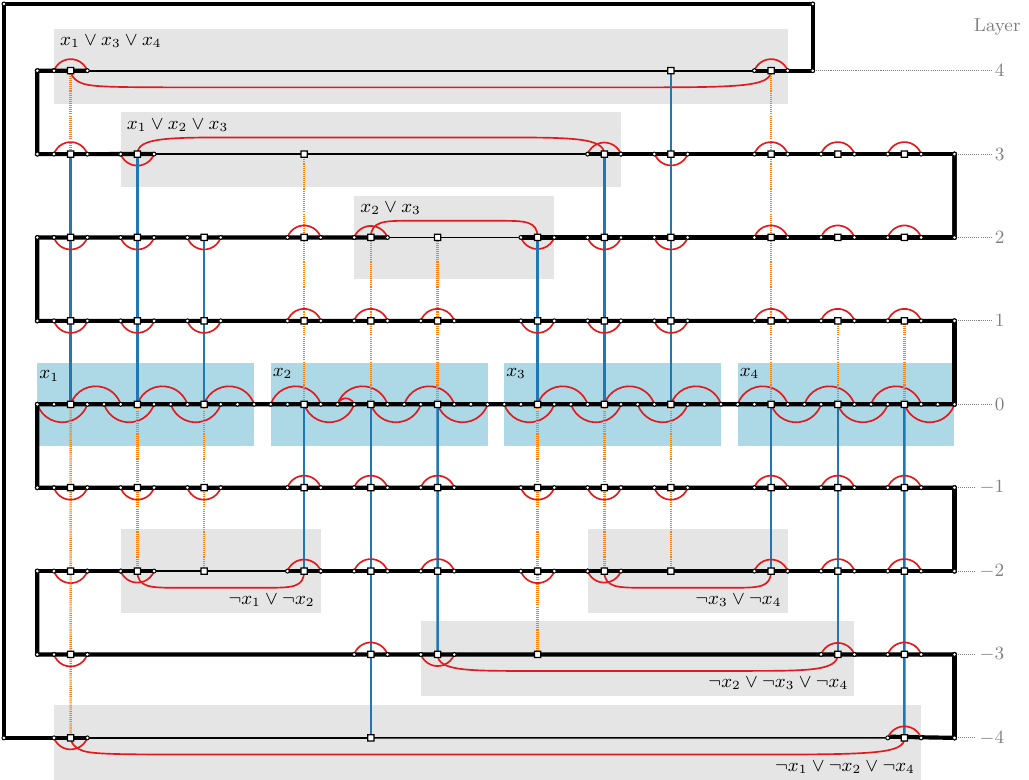}
    \caption{Solution (in red) of the \problemname instance coming from the graph given in~\Cref{fig:planar-monotone-3-sat} for $F$ being a matching.}
    \label{fig:plane-1-matching}
\end{figure}
 
\clearpage
\section{Conclusion}
We introduced the \problemnamek{k} problem and 
showed that it is \NP-complete for every $k\ge 1$ even when the given graph is biconnected and
the inserted edges form a path or matching.
We also presented a linear-time algorithm for \problemname when the 
given graph is triangulated.
This naturally raises the question if the triconnected case of
\problemname is also polynomial-time solvable.

\bibliographystyle{plainurl}
\bibliography{references_format}
\end{document}